\documentclass[a4paper,twocolumn,11pt, unpublished]{quantumarticle}
\pdfoutput=1

\pdfoutput=1

\usepackage[utf8]{inputenc}
\usepackage[english]{babel}
\usepackage[T1]{fontenc}
\usepackage{amsmath}
\usepackage{hyperref}
\usepackage[none]{hyphenat}

\usepackage{amssymb}

\usepackage{tikz}
\usepackage{lipsum}
\usepackage[final]{changes}
\usepackage{braket}
\usepackage{float}

\newtheorem{theorem}{Theorem}
\newtheorem{corollary}[theorem]{Corollary}

\begin{document}

\title{Born's Rule from Quantum Frequentism}
\author{Lionel Brits}
\affiliation{Department of Physics and Physical Oceanography, Memorial University of Newfoundland, NL A1C 5S7, St. John’s,
Canada}
\orcid{0000-0002-2445-2701}
\email{lbrits@mun.ca}

\maketitle

\begin{abstract}
    \added{Quantum theory has evolved from a set of provisional rules to an indispensable framework that underlies much of modern technology and infrastructure. Yet, after a century, Born's probability postulate remains at odds with the theory's unitary character. The problem stems from the linearity of the Schrödinger equation, as linear systems are insensitive to the magnitudes of their solutions' coefficients. If measurement is unitary, and thus linear, how can the frequency of an outcome depend on the magnitude of its amplitude? And if not, at what scale does unitarity break down?
    This question remains pressing, as the assumption of unitarity underlies both the design of large-scale fault-tolerant quantum devices, as well as our understanding of fundamental aspects of our universe, for example, the black hole information problem.} Proponents of the many-worlds interpretation have argued that Born’s rule is observed because events that violate it have vanishing norms, and so must be unphysical. However, this argument has only been made explicit for special cases involving infinitely many identical states. 
    In this paper we provide a generalized form of Born's rule applicable to measurements of arbitrarily-prepared factorizable states, and prove that such systems contain no histories that violate Born’s rule. \added{Our result therefore demonstrates that purely unitary evolution can co-exist with Born's rule under more relaxed conditions. In addition, we apply our generalized rule to the problem of quantum  circuit fidelity, and provide a single-shot alternative to cross-entropy benchmarking based on self-information rate. }
\end{abstract}

\section{Introduction}

Despite the predictive success of the Copenhagen formulation of quantum mechanics, it is not universally accepted as a complete description of nature~\cite{Schlosshauer_2013}. This is in part because it relies on a distinguished observer in order to make sense of the measurement process~\cite{Bell_1990}. This can either be seen as an inconsistency -- that physical interactions are unitary or non-unitary depending on
whether they are observed~\cite{Wigner1961} -- or merely a nuisance -- that only a single observer is actually necessary to bootstrap the measurement process, one
that can be pushed to the very edges of a system and then be forgotten~\cite{vonNeumann1955}. This issue has, in part, lead to
the development of the many worlds interpretation
(MWI)~\cite{Everett1956, DeWitt1973}, which aims to explain the role of the observer within the unitary framework of the
theory itself. This interpretation explains the apparent collapse of the wavefunction as a loss of coherence between environmental states corresponding to different measurement outcomes, so that local degrees of freedom seem to evolve irreversibly~\cite{Zurek_2003,joos2003decoherence}. However, the MWI suffers from its own minimalism -- having thrown out everything that is discontinuous and non-unitary, it does not seem to give a wholly satisfactory explanation for the origin of probability, and in particular, Born's rule. \added{If the measurement process is completely unitary, and therefore linear in the amplitudes of the wave-function, how can the probability of measuring a certain outcome depend on the magnitude of its amplitude? This question is not purely aesthetic. Efforts to observe quantum effects at mesoscopic scales -- such as in superconducting quantum devices~\cite{AbuGhanem_2025} and optomechanical resonators~\cite{kotler2021direct,bild2023schrodinger} -- as well as to describe more exotic environments, including black hole evaporation~\cite{raju2022lessons} and the early universe, demand a better understanding of the apparent transition between quantum and classical behavior.}

\added{Numerous attempts have been made to reconcile these two regimes, and to obtain Born's rule from unitarity measurement alone.} Some approaches  invoke arguments from symmetry~\cite{Zurek2005,lesovik2021derivation, hossenfelder2021derivation}, while others make use of decision-theoretic methods concentrated on the expectations of rational agents~\cite{Deutch1999, Wallace2010, Sebens_2018}. While compelling, many of these ideas center around a common question -- \textit{What should I expect \emph{if} probabilities were assigned to amplitudes?}~\cite{Gleason1957} -- and do not address the origin of the probability measure itself. In contrast, others argue that it is the behaviour of the wavefunction in the limit \added{of infinitely many indentical subsystems} that leads to the origin of a probability rule~\cite{Hartle1968, Farhi1989}, but efforts at extending this line of reasoning to more natural, non-identical systems have stalled. Our work builds upon the ideas of Hartle~\cite{Hartle1968} and Farhi, Goldstone, and Gutmann~\cite{Farhi1989}, who attempted to show the emergence of Born's rule in the infinite-$N$ limit using   the frequency operator formalism. We instead analyze the infinite-$N$ wavefunction directly, and offer a more general, constructive approach applicable to infinite, arbitrarily-prepared product states, showing how Born's probability rule arises from the structure of the universal wavefunction alone.

\section{The problem with (too) many worlds}

For definiteness, we will summarize the key ideas behind the MWI~\cite{Everett1956,DeWitt1973}. According to this interpretation, the state of any isolated system is at all times described by a state vector $\ket{\Psi}$ evolving unitarily according to the Schr\"odinger equation. In particular, the MWI deduces that if the universe is an isolated system, then it too must be described in this way. 
The key strength of the MWI is that it explains quite clearly how the stochastic nature of quantum mechanics comes about. Suppose that we divide our universe into a microscopic system $S$ being measured and an apparatus $A$ that is performing the measurement (the observer may be considered as part of the apparatus subsystem). Then a factorizable state $\ket{S_0} \otimes \ket{A_0}$ of the composite system will evolve into a superposition $\sum_i c_i \ket{S_i} \otimes \ket{A_i}$.  One may say that the universe in which the system and apparatus was in state $\ket{S_0} \otimes \ket{A_0}$ has branched into a set of universes each in the definite state $\ket{S_i} \otimes \ket{A_i}$. Within each universe, the observer sees a different measurement outcome (represented by $A_i$), despite having identical initial conditions, and since any particular observer has no way of knowing which branch they will find themselves in, their measurement outcome is completely unpredictable. We clarify that the universes described so far are merely arbitrary orthogonal decompositions of the universal wavefunction, and it is the task of the decoherence program to find emergent classical realities among these decompositions~\cite{Zurek1982,Joos1985,Zurek2003,joos2003decoherence}, something which we will not do here.

Despite the appeal of dealing with the observer as part of the quantum system, the MWI has been criticized for failing to account for the appearance of Born's rule in the measurement process~\cite{Kent_2010, Vaidman_SEP, vaidman2022many}. Since linearity puts every universe on an equal footing, it would suggest that they are all equally likely, including so-called ``maverick worlds'' in which Born's rule is grossly violated~\cite{DeWitt1973}. This is in contrast to what we actually observe, i.e., that outcome frequencies are proportional to the absolute squares of the coefficients $c_i$. \added{Hartle~\cite{Hartle1968} and Farhi, Goldstone, and Gutmann~\cite{Farhi1989} have argued that this result can be obtained in the limit of infinitely many identically prepared states (see also \cite{Graham1973-GRATMO-15}), where such maverick worlds have vanishing norm. To better understand the fate of these worlds then, we should begin with a minimal model of measurement in which Born's rule can be tested in a mathematically transparent way. }

\section{A model of measurement}
\added{As the probability postulate in it's most basic form is concerned with the relative frequencies of measurement outcomes,} let us imagine that a number of spin-$\tfrac{1}{2}$ particles are prepared in identical states equal to $\ket{\psi} = \alpha \ket{\downarrow} + \beta \ket{\uparrow} = \alpha \ket{1} + \beta \ket{0}$
after which the number of particles in one of the states, $\ket{\downarrow}$ say, is recorded by a particle counter. We may represent such an aparatus by a quantum circuit that counts the number of 1s in a set of identically prepared qubits (i.e., its Hamming weight) and records the result in a binary register consisting of some other previously initialized qubits.  
For the sake of clarity we will construct it from a sequence of unoptimized controlled-increment gates. This gate increments the target register  if the control qubit is set to $1$, and  leaves it unaltered otherwise. The full circuit is 
shown in Fig. \ref{fig:circuitcounter}.

\begin{figure}[h]
    \centering
    \includegraphics[width=8cm]{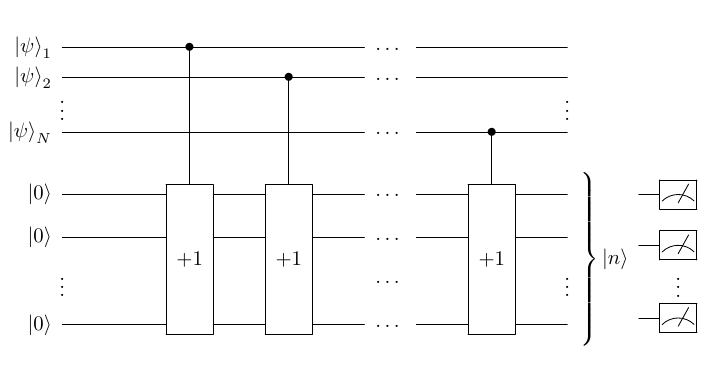}
    \caption{\added{A minimal realization of an ideal measurement apparatus in terms of a quantum circuit that records the number of 1s in its first $N$ inputs, which represent identical copies of the state $\ket{\psi}$. The register $\ket{n}$ represents the apparatus memory. Each $\fbox{+1}$ unitary represents a controlled-increment operation that adds $1$ to $n$ if the control register is set to $\ket{1}$. Before observation, the system remains in a superposition of states with $n$ between $0$ and $N$, with maximum amplitude near $n \approx N |\alpha|^2$.}}
    \label{fig:circuitcounter}
\end{figure}

To be sure, this circuit does not constitute a justification for Born's rule -- its purpose here is only to faithfully monitor the state-vector of the system throughout the operation of the detector: Prior to the first particle being recorded, our circuit may be considered to be in the state $\ket{\Psi_0} = \left(\alpha\ket {1} + \beta \ket {0}\right)^{\otimes N}  \ket{ 0}$, with the rightmost $\ket{0}$ representing the state of the counter. 
After $N$ particles, the system will be in the state

\begin{align}
\ket{ \Psi_N } &=  \sum_{n=1}^N \alpha^{n} \beta^{N-n} \sum_{h} \ket{ s_{n,h}}  \ket{ n},
\label{eqn:PsiN0}
\end{align}
where $\{ \ket{  s_{n,h}} \}$ is the set of subsequent $N$-particle states corresponding to the ${N\choose n}$ initial states which contain $n$ downward spins. 
Making use of the fact that the states $\ket{ s_{n,h}}$ are orthogonal, one finds,
\begin{align}
\label{eq:nupPsi}
\left\vert \left\langle n | \Psi_N \right\rangle \right\vert &=  \sqrt{ |\alpha |^{2n} |\beta|^{2(N-n)} {N \choose n}}.
\end{align}

The  term $|\alpha |^{2n} |\beta|^{2(N-n)} {N \choose n}$ may be recognized as the binomial distribution of the number of successes in a sequence of $N$ independent Bernoulli trials, each with success probability $p = |\alpha|^2$. This function has relative width $\frac{\sigma}{N} = \sqrt{\frac{p\, (1-p)}{N}}$ and mean $\mu = N |\alpha|^2$, so that after a large number of spins have been recorded, the wavefunction of the system will be in a superposition of states that are narrowly peaked around the expected relative frequency $|\alpha|^2$, decaying rapidly to zero in amplitude elsewhere. \added{While this result is consistent with Born's rule, since the most likely outcome is to observe a relative frequency near $|\alpha|^2$, it does not consistute a proof, as it relies on an existing probablistic interpretation of $\left\vert \braket{ n | \Psi_N } \right\vert^2$. On the other hand, this internal consistency is a hint that Born's rule is deeply connected to the structure of the wavefunction itself.}

\added{In order to illustrate the behaviour of our measurement model in a physically concrete yet mathematically transparent form,  we implemented a version of this circuit for $N=4$ (see Fig.~\ref{fig:circuitcounter2} below)} 
on the 133-qubit IBM Heron processor \texttt{ibm\_torino} with 4 qubits  prepared in the state \(\frac{1}{\sqrt{2}} \left( \ket{0} + \ket{1}\right)\). This circuit was chosen for its simplicity, as the device's coherence budget did not allow for counting beyond $N=4$. \added{Additionally, post-selection was performed to mitigate single-qubit errors, whereby shots with erroneous checksums were discarded (a total of $8020$ out of $30000$). A histogram of measurement readouts of the register $\ket{n}$ is shown in Fig.~\ref{fig:histogram}.}

\begin{figure}[H]
    \centering
    \includegraphics{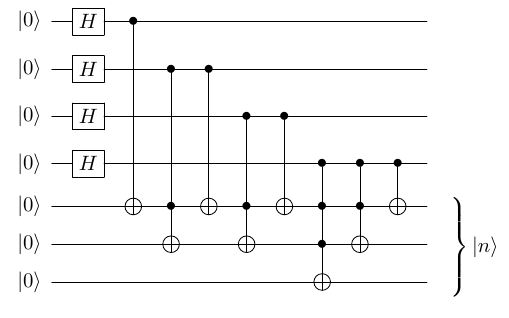}
    \caption{\added{Minimal test of ideal measurement apparatus for $N=4$ qubits (see Fig. \ref{fig:circuitcounter}) executed on the IBM Heron processor \texttt{ibm\_torino}. Circuit counts the number of qubits in the $\ket{1}$ state among $4$ copies of the state $\ket{\psi} = U_H \ket{0}$ and stores the result in the 3-qubit binary register $\ket{n}$. A  distribution of observed values of $n$ is shown in Fig. \ref{fig:histogram}.
    }}
    \label{fig:circuitcounter2}
\end{figure}

\begin{figure}[H]
    \centering
    \includegraphics[width=0.5\textwidth]{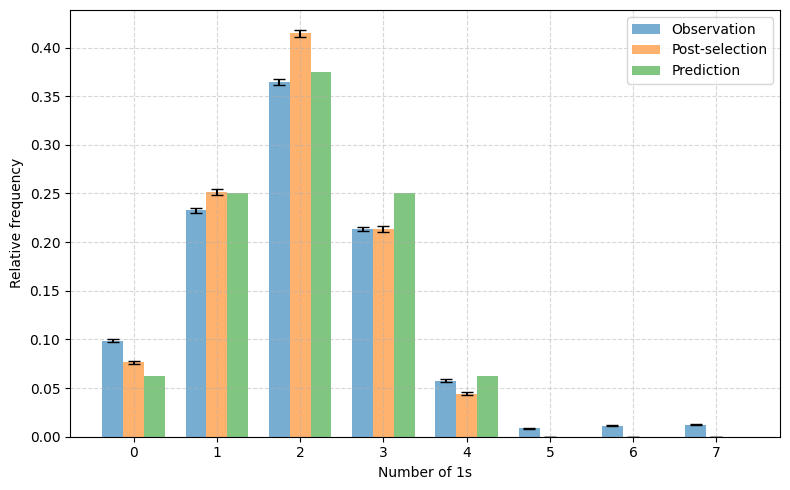}
    \caption{\added{Distribution of measurement frequencies from 30k shots of the 4-qubit counting circuit shown in Fig.~\ref{fig:circuitcounter2} executed on the 133-qubit IBM Heron processor \texttt{ibm\_torino}. Each bin contains: raw frequencies (left); post-selected frequencies after discarding shots that failed single-qubit checksum (center); theoretical binomial distribution (right). Error bars on observed data represent one standard deviation from finite-shot statistical uncertainty. Relative freqency of 1s supports a binomial distribution with $N=4$ and $p=0.5$ with deviations consistent with circuit noise and decoherence. 
    }}
    \label{fig:histogram}
\end{figure}

\added{As can be seen from the histogram in Fig.~\ref{fig:histogram}, the observed frequency of qubits in the $\ket{1}$ state is approximately symmetrically distributed near $n=2$, and support a binomial  distribution with $N=4$ and $p=\tfrac{1}{2}$ (compare with Equation~\eqref{eq:nupPsi}). While the measurements were largely limited by decoherence and circuit noise, the results are consistent with rule Born's applied to the state \(\frac{1}{\sqrt{2}} \left( \ket{0} + \ket{1}\right)\).   Nevertheless, since the circuit is reversible (up to decoherence effects), the system must remain in a superposition of different counter states $\ket{n}$ prior to its final observation.}

This experiment therefore resembles a sort of robotic ``Wigners's friend''~\cite{Everett1956,Wigner1961,Deutsch1985}, in which an observer (the particle counter) is monitored by a second observer (the experimenter) -- a scenario which cleanly illustrates the sort of contradictions inherent to the Copenhagen interpretation~\cite{Frauchiger_2018}. From the point of view of the counter, the system contains a definite number of \(\ket{\downarrow}\) states, whereas from the point of view of the experimenter, the system as a whole appears to be in a superposition of states in which every possible value of $n$ has been registered. As  $N$ is increased, the majority of these states correspond to frequencies that deviate significantly from Born's rule, i.e., $n$ is far from $N|\alpha|^2$. Such outcomes are very unlikely, however, and most observations will be centered around the expected value $\braket{n} = N|\alpha|^2$. Of course, this fact just pushes the observation process one level further, from the particle counter, to the experimenter.

At this point the special role
of an observer seems indispensable for obtaining Born's rule, but this can only
be true if they are to be endowed with some non-unitary quality~\cite{Everett1956}. While this possibility has yet to be ruled out experimentally, recent demonstrations of quantum superposition on both mesoscopic and macroscopic~\cite{OConnell2010, Fein2019} scales make this argument hard to accept. An alternative conclusion, originally put forth by Hartle~\cite{Hartle1968}, and Farhi, Goldstone, and Gutmann~\cite{Farhi1989}, is that maverick worlds are nonphysical, since they are assumed to have vanishing norm (and not simply vanishing likelihood). These works showed that a wavefunction containing identically-prepared states becomes an exact eigenstate of the relative frequency operator as $N$ is taken to infinity, as the component of the wave-function containing histories that deviate from Born's rule vanishes. This situation is essentially the infinite limit of the circuit discussed above. Since a well-formed Hilbert space must contain only a single vector of zero norm (i.e., the null vector), it seems justifiable to treat such histories as artefacts of representation rather than physical possibilities. In the next section we will present this argument in detail,  before moving on to its generalization to arbitrary product states.

\section{Special case: infinite identical states}

Careful attention to Equation~\eqref{eqn:PsiN0} shows that the counter has so far served only to organize the superposition into states with definite numbers of downward spins. We may therefore focus our attention on the initial state vector of the particles alone. First, we rewrite the state in terms of the computational basis, i.e. 
\begin{equation}
    \label{eqn:PsiN}
\ket{\Psi_N} = \sum_{x_1 x_2\dots x_N} c_{x_1 x_2 \dots x_N   } \ket{x_1}\ket{x_2}\dots \ket{x_N},
\end{equation}
where $x_i \in \{0,1\}$ and $c_{x_1 x_2 \dots x_N} =  \alpha^{n} \beta^{N-n}$ with $n = \sum x_i$. Each state of the form $\ket{x_1}\ket{x_2}\dots \ket{x_N}$ is therefore an eigenstate of the operator $\hat n = \sum_i \hat x_i$, meaning that we can associate a definite counter value $n$ to each set of measurement outcomes $(x_1, x_2, \dots, x_N)$. We can show that, as $N\to\infty$, the amplitude of the wave-function (\ref{eqn:PsiN}) becomes negligible away from $n = N |\alpha|^2$ (and in fact vanishes). However, in this limit, the vector space $\mathcal{H} = \mathbb{C}^{2^N}$ spanned by the basis vectors $\ket{x_1}\ket{x_2}\dots$ becomes non-separable, so that we have to be careful to maintain a well-behaved inner product structure~\cite{vonNeumann1939}.  Let us define a maverick world to be one in which the observed relative frequency $f = \frac{1}{N} \sum x_i$ differs from the expected frequency $\mathrm{E}[f] = |\alpha|^2$ by some finite positive error $\epsilon$, i.e., 

\begin{equation}
\left| f - |\alpha|^2\right| > \epsilon.\label{eqn:xbarepsilon}
\end{equation}
(Here and elsewhere expectation values will always mean those computed according to Born's rule). We can then decompose $\ket{\Psi_N}$ into the projection $\ket{\Psi_N}_\mathrm{M(averick)}$ containing all maverick worlds, as well as the projection $\ket{\Psi_N}_\mathrm{B(orn)}$ containing all regular, or Born worlds, so that $\ket{\Psi_N} = \ket{\Psi_N}_\mathrm{B} + \ket{\Psi_N}_\mathrm{M}$. To find $\left\lVert\ket{\Psi_N}_\mathrm{M}\right\rVert$, we would need to evaluate 
\begin{equation}
\left\lVert\ket{\Psi_N}_\mathrm{M}\right\rVert^2 = \sum_{\left| f - |\alpha|^2\right|\, >\, \epsilon } |c_{x_1 x_2 \dots x_N   }|^2.
\end{equation}
We can however place an upper bound on this quantity. Because $n$ is the sum of independent random variables taking the values $\{0,1\}$ we use Hoeffding's inequality~\cite{Hoeffding1963} to find
\begin{equation}
\left\lVert\ket{\Psi_N}_\mathrm{M}\right\rVert^2 \le 2 \, e^{-2 \epsilon^2 N}.
\end{equation}
It follows that $\lim_{N\to\infty} \left\lVert\ket{\Psi_N}_\mathrm{M}\right\rVert^2 = 0$ for any finite $\epsilon$, and therefore that $\ket{\Psi_\infty}$ differs from $\ket{\Psi_\infty}_\mathrm{B}$ by a quantity of zero norm. Such spurious vectors are not proper elements of the Hilbert space, and must be removed in order to maintain a positive definite
inner product~\cite{vonNeumann1939}. That such states are orthogonal to (and therefore decoupled from) physical states can be seen directly from the Cauchy--Schwarz inequality,
\begin{equation}
\left| \braket{a|b} \right| \le \braket{a|a} \braket{b|b} ,\quad \forall \ket{a}, \ket{b} \in \mathcal{H},
\end{equation}
since if $\braket{a|a}=0$ then the overlap of $\ket{a}$ with any other vector is zero.
 It is therefore necessary that we identify those elements of $\mathcal{H}$ that differ by elements of $\mathcal{H}_0$, the subspace of zero norm states. Our Hilbert space is then the quotient space denoted by $\mathcal{H}/\mathcal{H}_0$, and we may consider $\ket{\Psi_\infty}$ and $\ket{\Psi_\infty}_B$ to represent the same physical state, i.e., \mbox{$\ket{\Psi_\infty} =  \ket{\Psi_\infty}_B$}. As for the fate of the spurious sequences represented by \(\ket{\Psi_\infty}_M\), it  may be more appropriate to think of them as finite-\(N\) artefacts. We conclude that in the $N\to\infty$ limit the universal wavefunction contains only those sequences that agree with Born's rule.


An immediate obstacle to the frequentist program is that one does not perform measurements on systems containing infinitely many identically prepared particles, and this was already pointed out by Farhi,
Goldstone, and Gutmann ~\cite{Farhi1989} (see also~\cite{Caves2005}). Likewise, for finite \(N\), all frequencies of events have non-zero amplitudes, and consequently non-zero probabilities of occurring, in which case the question of whether a sequence obeys Born's rule seems ill-defined.
Everett and others have ultimately argued that the only way to produce the correct Born probabilities for finite measurements is to assign a probability measure on the universal Hilbert space, so that maverick worlds are never observed simply because one is very unlikely to find oneself in such a world. However, this apparent resolution is circular, and is at odds with the notion of the wavefunction as a complete description of the system~\cite{Kent_2010, Vaidman_SEP, vaidman2022many}. If quantum mechanics is to be self consistent, then the structure of the wavefunction itself must account for the appearance of Born's rule. In the next section, we will see that we can in fact generalize the above derivation to measurements on finite, arbitrarily prepared subsystems, as long as those subsystems are embedded in an infinite environment.

\section{General result: arbitrary product states}

The failure of the frequentist approach stems from its inability to account for the probability of encountering certain finite sequences of outcomes without invoking Born's rule at some point. A way out of this problem is to realize that any experiment performed on a finite multi-particle state such as $\ket{\psi} \ket{\psi}...\ket{\psi}$ must take place inside a Hilbert space containing also all the particles in the environment. Aguirre and Tegmark~\cite{Aguirre2011} have argued that in an infinite, statistically uniform cosmological model the laboratory state $\ket{\psi}$ must be replicated infinitely many times throughout the universe, thereby realizing the  ``fictitious'' infinite ensemble needed to derive Born's rule. (More specifically, the authors of impose the stronger condition that both system plus experimenter be replicated infinitely many times, although this does not seem necessary.) However, this argument rests on some knowledge of the distribution of states that make up the universal wavefunction, and does not account for states that come arbitrarily close to, but never equal $\ket{\psi}$. As we will show, the replica condition of ~\cite{Aguirre2011} is actually unnecessary, so that we need only consider states of the form
\begin{equation}
\ket{\Omega} = ...\ket{\varphi_{-2}}\ket{\varphi_{-1}}\,\Big(\ket{\psi}\ket{\psi}...\ket{\psi}\Big)\,\ket{\varphi_{1}}\ket{\varphi_{2}}...,
\end{equation}
where $\{\ket{\varphi_i}\}$ now represent arbitrary environmental component states. Since there is no real distinction between the system and the environment, it is convenient to treat all degrees of freedom on an equal footing by absorbing the multi-particle state $\ket{\psi}\ket{\psi} ...\ket{\psi}$ into the environmental degrees of freedom, letting
\begin{equation}
\ket{\Omega_N} = \ket{\phi_1}\ket{\phi_2}\dots \ket{\phi_N}.
\label{eqn:Omega_N}
\end{equation}
Of course, the reader may still imagine that some finite subset of components of  $\ket{\Omega_N}$ are identically prepared, but this is not strictly necessary. It is important to note, however,  that the form of $\ket{\Omega_N}$ assumes that the universal wavefunction is factorizable into non-entangled subsystems (at least in some basis) and excludes strongly entangled states. Whether Born's rule arises in such cases remains an open question, and we leave its investigation for future work.

What we need, then, is a way to verify whether a measurement on an arbitrary state such as $\ket{\Omega_N}$ obeys Born's rule, i.e., an analogue of Equation~\eqref{eqn:xbarepsilon} for arbitrary sequences. In terms of the computational basis, this state may again be written as
\begin{equation}
\ket{\Omega_N} = \sum_{x_1 x_2\dots x_N} c_{x_1 x_2 \dots x_N   } \ket{x_1}\ket{x_2}\dots \ket{x_N}.
\end{equation}
As we have argued, there can be no strict condition placed on the outcome of any finite subset of measurement outcomes $(x_i)$ (provided that $p_i$ is neither $0$ or $1$). \added{Instead, we will take an information theoretic approach, and use the \textit{self-information rate}~\cite{Shannon1948,Cover2006} of the sequence $(x_i)$ as our metric. This value, also known as the \textit{surprisal rate}, is given by
\begin{equation}
-\frac{1}{N}\ln p(x_1,x_2,\dots,x_N) 
\end{equation}
and quantifies how unexpected a sequence $(x_i)$ is given a joint probability distribution $p(x_1, x_2,\dots, x_N)$. For a single variable $x$ described by a probability distribution $p(x)$ we see that this value ranges from $0$ when the outcome is guaranteed ($p(x) = 1$), to $\infty$ when the outcome is forbidden ($p(x) = 0$). } More generally, it allows us to  define a \textit{typical sequence} to be a sequence $(x_i)$ such that
\begin{equation}
\label{eq:deftypical}
\left| -\frac {1}{N}\ln p(x_1,x_2,\dots,x_N) - H \right|\le \epsilon,
\end{equation}
for some  value $\epsilon > 0$, where
\begin{equation}
H= \frac{1}{N} \mathrm{E}\left[-\ln p(x_1,x_2,\dots,x_N)\right],
\end{equation}
is the Shannon entropy rate\footnote{Not to be confused with the von Neumann entropy, which is zero in this case.} of the distribution $p(x_1,x_2,\dots,x_N)$~\cite{Shannon1948}. Typical sequences are representative sequences of a probability distribution, an idea made precise  by the asymptotic equipartition property (AEP)~\cite{Shannon1948,Cover2006}, which states that
\begin{equation}
\lim_{N \to \infty} \mathrm{Pr}\left[ \left| -\frac {1}{N}\ln p(x_1,\dots,x_N) - H(x) \right| > \epsilon\right] = 0
\label{eqn:AEP}
\end{equation}
That is, as $N$ becomes large, the self-information rate of a sequence chosen at random from the distribution $p(x_1,x_2,\dots,x_N)$ tends towards $H$, which follows directly from the weak law of large numbers applied to the quantity $ -\ln p(x_1,x_2,\dots,x_N)$. We can therefore partition the set of all possible sequences $(x_i)$ into two sets: a typical set, in which every element has probability $p(x_1,x_2,\dots,x_N) \approx e^{-N H}$, and a non-typical set, containing all other sequences. The AEP tells us that, in the $N \to \infty$ limit, the probability of randomly selecting an element that belongs to the non-typical set is zero. 

Equation~\eqref{eq:deftypical} therefore provides a precise operational definition of Born's rule in the case of arbitrarily prepared states: Given a state $\ket{\Omega_N}$ and the  associated quantity $p(x_1,x_2,\dots,x_N) = \left| \braket{x_1,x_2,\dots,x_N | \Omega_N  }\right|^2$, we say that a sequence of experimental outcomes $\{x_1, x_2,\dots,x_N\}$ obeys Born's rule if it satisfies Equation~\eqref{eq:deftypical} for some chosen threshold $\epsilon$. Importantly, we introduce the AEP and the probability measure it describes \emph{not as an assumption}, but as a means to \emph{distill} the statement that we aim to prove.

To see that this definition agrees with our intuition, note that if we have $N'$ repetitions of the state $\sum_k a_k \ket{k}$ then the inequality in Equation~\eqref{eq:deftypical} will contain the term
\begin{equation}
\sum_i \ln \left|a_{x_i}\right|^2 - N' \sum_{k} \left|a_{k}\right|^2 \ln \left|a_{k}\right|^2,
\end{equation}
which attains a minimum when $\ln\left|a_k\right|^2$ occurs roughly $N' \left|a_k\right|^2$ times in the first sum, or equivalently, when there are roughly $N' \left|a_k\right|^2$ particles in the state $\ket{k}$.

Now that we have a robust definition of Born's rule for arbitrarily prepared product states, we can proceed to show that only sequences that obey this rule persist in the $N\to\infty$ limit, that is, that the state \mbox{$\ket{\Omega_N} = \ket{\Omega_N}_B + \ket{\Omega_N}_M$} contains no non-typical sequences in the $N\to\infty$ limit. Fortuitously, the squared norm $\left\lVert\ket{\Omega_N}_\mathrm{M}\right\rVert^2$ is precisely equal to the quantity $\mathrm{Pr}\left[ \left| -\tfrac {1}{N}\ln p(x_1,x_2,\dots,x_N) - H \right| > \epsilon\right]$ already considered. Therefore, by Equation~\eqref{eqn:AEP}, we may conclude that $\lim_{N\to\infty} \left\lVert\ket{\Omega_N}_\mathrm{M}\right\rVert^2 = 0$ and that $\ket{\Omega_\infty} = \ket{\Omega_\infty}_B$ in general. This result is somewhat subtle: the condition we used to define Born sequences turns out to also guarantee that only such sequences survive in the large-$N$ limit. The structure of the universal wavefunction itself therefore ensures the emergence of Born’s rule, without requiring any external probability measure.

Here we pause to scrutinize this result. Although it may seem that we appealed to a probabilistic argument, we stress that the AEP and the weak law of large numbers from which it is derived are purely algebraic statements, and Equation~\eqref{eqn:AEP} may be applied directly to the squared amplitudes $ \left| c_{x_1 x_2 \dots x_N}\right|^2$ without invoking the probability postulate. In other words, we do not discard maverick worlds because they are improbable, but rather because they are unphysical (not elements of the Hilbert space proper). Since our argument hinges on the fact that the AEP is a non-probabilistic statement, it is worth demonstrating that explicitly. We begin by recalling an important inequality (see appendix):

\begin{theorem}[Chebyshev's Inequality] \label{thm:cheb} Given an arbitrary state $\ket{\Psi}$ and a Hermitian operator $\hat Y$, let $\hat P_{|Y-\mu_Y| > \epsilon}$ be the projection operator that preserves states for which $\left|Y-\bra{\Psi} \hat Y \ket{\Psi}\right| > \epsilon$ for some $\epsilon > 0$. Then
    \begin{equation}
    \left\lVert \hat P_{|Y-\mu_Y| >\epsilon }\ket{\Psi}\right\rVert ^2 \le \frac{\mathrm{Var}(Y)}{\epsilon^2},
    \end{equation}
    where $\mathrm{Var}(Y) = \bra{\Psi} (\hat Y- \bra{\Psi} \hat Y \ket{\Psi})^2 \ket{\Psi}$.
    \end{theorem}
Phrased in this form, Chebyshev's inequality is independent of any probabilistic interpretation, and may be used to derive the AEP by applying it to $\ket{\Omega_N}$, letting $Y(x_1,x_2,\dots,x_N) = -\frac{1}{N} \ln \left| c_{x_1  \dots x_N   }\right|^2$. 
\begin{widetext}
Then
\begin{align}
\left\lVert \hat P_{\left|- \frac{1}{N} \ln \left|c_{x_1  \dots x_N   }\right|^2 - H(x)\right| > \epsilon }\ket{\Omega_N}\right\rVert ^2 &\le \frac{\mathrm{Var}( \ln \left|c_{x_1 \dots x_N   }\right|^2)}{N^2 \epsilon^2}.
\end{align}
Since $\ket{\Omega_N} = \ket{\phi_1}\ket{\phi_2}\dots \ket{\phi_N}$ the quantity $\left|c_{x_1 \dots x_N   }\right|^2$ can be factorized into the form $|c_1(x_1)|^2\dots |c_N(x_N)|^2$ so that
\begin{align}
\left\lVert \hat P_{\left|- \frac{1}{N} \ln \left|c_{x_1  \dots x_N   }\right|^2 - H(x)\right| > \epsilon }\ket{\Omega_N}\right\rVert ^2 &\le \frac{\sum_i \mathrm{Var}( \ln |c_i(x_i)|^2)}{N^2 \epsilon^2}.
\end{align}
\end{widetext}
Provided that $\mathrm{Var}(\ln|c_i(x_i)|^2) < M$ for all $i$ (which may be ensured by standard regularization techniques) we may write
\begin{align}
\left\lVert \hat P_{\left|- \frac{1}{N} \ln \left|c_{x_1  \dots x_N   }\right|^2 - H(x)\right| > \epsilon }\ket{\Omega_N}\right\rVert ^2 &\le \frac{M}{N \epsilon^2}.
\label{eqn:AEP2}
\end{align}
This equation places an upper bound on $\left\lVert\ket{\Omega_N}_\mathrm{M}\right\rVert^2$ without any reference to probability, and to allows us conclude that $\lim_{N\to\infty} \left\lVert\ket{\Omega_N}_\mathrm{M}\right\rVert^2 = 0$ and that $\ket{\Omega_\infty} = \ket{\Omega_\infty}_B$ as claimed.

It may seem objectionable that we are discarding maverick worlds on the basis that $\ket{\Omega_\infty}_M=0$ when the normalization of the wavefunction requires that \textit{every} sequence has zero amplitude (in the infinite limit). However, this requirement is actually unjustified, since we cannot subject the universal wavefunction to repeated external measurement, and without any probablistic interpretation there is no reason to demand that $\left\Vert\ket{\Omega_\infty}\right\Vert^2 = 1$. \added{Without this restriction, we may view $\ket{\Omega_\infty}$ as a formal sum over states containing almost entirely typical sequences (along with an exponentially suppressed number of non-typical states with greater individual amplitudes).}

The fact that we observe Born's rule therefore stems from a rather surprising place: From a practical point of view, choosing a sequence $(x_1, x_2, \dots)$ at random from the distribution $p(x_1,x_2,\dots)$ is \textit{indistinguishable} from choosing among the set of typical sequences. Since $\ket{\Omega_\infty}$ consists almost entirely of typical sequences , our observations must appear to be governed by the distribution $p(x_1,x_2,\dots)$ as obtained from Born's rule. \added{This establishes that the probability postulate is guaranteed by unitary evolution provided that the universal wavefunction exhibits the structure required by the AEP.}

\section{Application to quantum benchmarking}

The self-information rate \( -\frac{1}{N} \log p(x_1, x_2, \dots) \) is not only key to understanding the emergence of Born’s rule in infinite systems, but is also a potentially useful metric for circuit fidelity in quantum computing -- provided the true entropy rate \( H(p) \) of the target distribution can be estimated. While metrics such as cross-entropy benchmarking (XEB)~\cite{GoogleSupremacy2019} are already employed, the difference between the self-information rate and the entropy rate effectively \emph{defines} the fidelity of a given data sequence to the target distribution, and is therefore \added{a potential alternative to other metrics, even at finite $N$.}
\added{As a means to experimentally validate this, as well as to gain a better understanding of the behaviour of self-information for ``large'' quantum systems}, we prepared an ensemble of 120 qubits in a pseudo-randomly distributed (i.e., reproducible) product state on the IBM Eagle processor \texttt{ibm\_sherbrooke}, and recorded  measurements over 4096 shots (See \mbox{\url{https://github.com/lbrits/born}}).  As can be seen from the histogram in Fig.~\ref{fig:entropycompare}, the self-information rates for each of these measurements (left) are narrowly clustered around the expected value for that state \added{(as predicted by the AEP)}. \added{In order to compare the self-information rates of measurements performed on the prepared states to that of sequences obtained by pure chance (or those which are saturated by noise), } we also computed the  rates of 4096 pseudo-random bitstreams relative to the same states (right). \added{As can be seen, the overlap between the two distributions is negligible, demonstrating the value of self-information rate for distinguishing the high-fidelity measurement outcomes from noise-dominated ones. }

\begin{figure}[H]
    \centering
    \includegraphics[width=0.5\textwidth]{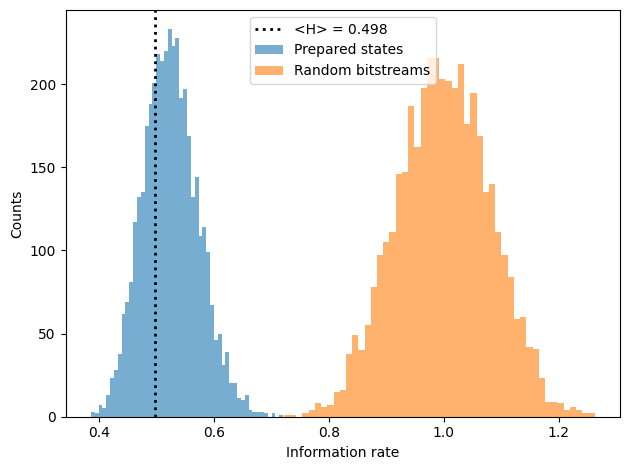}
    \caption{Histogram of observed self-information rates for 4096 measurements of 120 pseudo-randomly prepared qubits on the IBM Eagle processor \texttt{ibm\_sherbrooke} showing clustering around predicted entropy rate (left).  Self-information rates for 4096 uniformly random bitstreams of the same length, showing deviation from Born rule consistency (right).
    }
    \label{fig:entropycompare}

\end{figure}
Note that while we used a large number of shots to illustrate this point, self-information rate is calculated per-shot, and therefore is applicable to single-shot measurements. \added{This is in contrast to XEB, which requires many shots -- typically on the order of $10^6$ -- in order to estimate the empirical distribution. Self-information rate is therefore a potential single-shot, or few-shot alternative to XEB, at least in preparations that are compatible with the AEP. }

\added{It is also illustrative to give a general estimate of the entropy rate for randomly prepared states such as the one we have just considered. }
Given an ensemble of product states on $N$ qubits sampled independently from the Bloch sphere, i.e., the Haar measure on $SU(2)$, the probability $p$ of measuring the state $\ket{\uparrow}$ is uniformly distributed over $[0,1]$. The expected value $H(p) = -p \ln p - (1 - p) \ln (1 - p)$ can be computed directly:
\[
\mathbb{E}_{p}[H(p)] = \int_0^1 H(p) \, dp  = \frac{1}{2}.
\]
For a large number of qubits, we would expect the single-shot self-information rate to approach $\tfrac{1}{2}$. 
As we can see from Fig.~\ref{fig:entropycompare}, our pseudo-randomly prepared ensemble predicted an entropy rate very close to this ideal value, with $\braket{H} = 0.498$. This calculation can also be extended to the case of random states drawn uniformly from the full $2^N$-dimensional Hilbert space~\cite{Boixo_2018}, however, this is beyond the scope of this paper.

\section{Concluding remarks}
The importance of environmental degrees of freedom in obtaining Born's rule was already recognized by Zurek~\cite{Zurek1982} in the context of environmentally induced decoherence. However, in our result the environment plays a non-dynamical role, serving to define, and then get rid of maverick worlds, a sort of Mach's principle for quantum states.  Therefore, while system and environmental degrees of freedom are \textit{a priori} independent, they must be considered together as parts of a single quantum state. 
Our result shows that if one allows for systems with infinitely many degrees of freedom to exist, then Born's rule arises from the theory automatically. However, we do not impose any \textit{ad hoc} measure on the Hilbert space in order to achieve this result. Instead, we note that in order for the inner product to be positive definite, all vectors of zero norm must be identified with the zero vector. Thus the physical Hilbert space is not $\mathcal{H}$ but $\mathcal{H}/\mathcal{H}_0$, in which $\ket{\Omega_\infty}-\ket{\Omega_\infty}_B$ is identically zero.

It is also worth noting that the argument presented in this paper works strictly in the limit $N\to\infty$, and fails for any finite value of $N$, where maverick states are in the majority. If we were to assume that one is equally likely to find oneself in any of a finite number of histories, then we would not expect to observe Born's rule. Thus the statistical behaviour of the finite system does not approach that of the infinite system continuously. (Fig.~\ref{fig:entropycompare} does not contradict this fact, since for both datasets the number of measurements was fixed at 4096 for comparison purposes.) 

To fully account for Born's rule implicitly assumes that our universe contains infinitely many degrees of freedom (rather than some arbitrarily large number, as is often done). While this does not seem to be an overly objectionable assumption to us, we can turn this reasoning around and view Born's rule instead as an experimental validation of this possibility. \added{That is, the observation of Born's rule supports the idea that the universe contains infinitely many degrees of freedom, and assuming physical bounds on thermodynamic entropy~\cite{Bekenstein1981}, is therefore of infinite extent.}

Finally, while the present analysis applies to infinite product-state preparations, the fate of Born's rule in highly entangled or critical phases remains an open question. \added{Precisely which structural features, such as isotropy or ergodicity, the emergence of Born's rule hinges upon is a subject for future work.}

\section{Acknowledgments}

The author thanks Conor Stokes and Stephanie Curnoe for helpful discussions, and his advisor, James LeBlanc, for valuable guidance.

\bibliographystyle{quantum}
\bibliography{bornquantum}

\onecolumn

\appendix
\pagebreak
\section{Appendix}

\renewcommand{\thetheorem}{A.\arabic{theorem}}
\setcounter{theorem}{0}

The following are standard results rephrased in linear-algebraic terms to emphasize freedom from probabilistic interpretation.

\begin{theorem}[Markov's Inequality] \label{thm:markov}
Given an arbitrary state $\ket{\Psi}$ and a non-negative Hermitian operator $\hat f$, let $\hat P_{f > a}$ and $\hat P_{f \le a}$ be projection operators that separate those states for which $f > a$ from those for which $f \le  a$ for some $a>0$. Then
\begin{equation}
\left\lVert \hat P_{f > a}\ket{\Psi}\right\rVert ^2 \le \frac{1}{a} \bra{\Psi} \hat f\ket{\Psi}.
\end{equation}
\end{theorem}

\textit{Proof.}
Without loss of generality, we may consider a basis that diagonalizes $\hat{f}$, i.e., let 
\begin{equation}
\ket{\Psi} = \sum_{x_1 x_2\dots x_N} \Psi_{x_1 x_2 \dots x_N   } \ket{x_1}\ket{x_2}\dots \ket{x_N},
\end{equation}
such that
\begin{equation}
\hat f  \ket{x_1}\ket{x_2}\dots \ket{x_N} = f(x_1,x_2, \dots, x_N) \ket{x_1}\ket{x_2}\dots \ket{x_N}.
\end{equation}
Then
\begin{align*}
\bra{\Psi} \hat f \ket{\Psi} &= \bra{\Psi}\hat P_{f \le a} \hat f\hat P_{f \le a}\ket{\Psi} + \bra{\Psi} \hat P_{f > a}\hat f\hat P_{f > a} \ket{\Psi},\\
&\ge\bra{\Psi}\hat P_{f > a} \hat f \hat P_{f > a}\ket{\Psi} ,\\
&\ge a \bra{\Psi}\hat P_{f > a} \hat P_{f > a}\ket{\Psi}.
\end{align*}
Since $a > 0$, 
\begin{equation}
\left\lVert \hat P_{f > a}\ket{\Psi}\right\rVert ^2 \le \frac{1}{a} \bra{\Psi} \hat f\ket{\Psi}.
\end{equation}

\begin{corollary}[Chebyshev's Inequality]
Theorem~\ref{thm:cheb} in the main text follows immediately from Markov's inequality taking $\hat f = (\hat Y- \mu_Y)^2$ where $\mu_Y = \bra{\Psi} \hat Y \ket{\Psi}$ and $a = \epsilon^2$. 
\end{corollary}

\end{document}